\documentclass[11pt]{article}

\usepackage{epsf,amsmath,amsfonts,amssymb,amsthm,graphicx,mathrsfs,latexsym,enumerate,leftidx}
\usepackage{epsfig} 
\usepackage[text={7.0in,9.5in},centering]{geometry} 
\usepackage{multirow}
\usepackage{color}

\newtheorem{prop}{Proposition}[section]
\newtheorem{de}{Definition}[section]
\newtheorem{cor}{Corollary}[section]
\newtheorem{thm}{Theorem}[section]

\newtheorem{re}{Remark}[section]
\newtheorem{ass}{Assumption}[section]

\newtheorem*{acknowledgement*}{Acknowledgement}

\makeatletter
\newcommand{\vast}{\bBigg@{4}}
\newcommand{\Vast}{\bBigg@{5}}
\makeatother

\newcommand{\sm}{\scalebox{0.75}[1.0]{\( - \)}}

\begin{document}

\title{Stochastic Replicator Dynamics Subject to Markovian Switching }
\author{Andrew Vlasic \\ Department of Mathematics and Statistics \\ Queen's University }
\date{}

\maketitle

\begin{abstract}
Population dynamics are often subject to random independent changes in the environment. For the two strategy stochastic replicator dynamic, we assume that stochastic changes in the environment replace the payoffs and variance. This is modeled by a continuous time Markov chain in a finite atom space. We establish conditions for this dynamic to have an analogous characterization of the long-run behavior to that of the deterministic dynamic. To create intuition, we first consider the case when the Markov chain has two states. A very natural extension to the general finite state space of the Markov chain will be given. 
\end{abstract}


\section{Introduction}
Stochastic environments where independent external forces change the dynamic of the system are common in biological and economic settings \cite{KL05, H89, BW11, JW11, C09, GGMP12, TDHS06, HFV13, FL12, KKBL05, aidoo2002protective}.  An example to illustrate a complete and state-independent change in the dynamic is sickle-cell anemia \cite{aidoo2002protective, allison1954protection}. Being a carrier for sickle-cell lowers an individual's fitness, however, during malaria outbreaks, since sickle-cell carriers have immunity, which increases their fitness. The random event of malaria outbreaks may be described as a continuous-time Markov chain that is independent of the population dynamic,  yet changes the population dynamic. 

Antibiotics affecting microbial populations is another example of a population subjected to an independent stochastic environment \cite{AH05}. The authors discuss the affects of antibiotics to bacteria, such as \textit{Escherichia coli} and \textit{Salmonella enterica}, and what persistent environment is needed to support the type that is antibiotic resistant.

Although the motivation for our model is mostly biological, there is also a relationship with economics. A commonly used tool in economics is Markovian switching \cite{H89, BW11, JW11, C09}. This includes modeling business cycles and GDP growth, electricity spot price models, and interest rates, all of which are agents working in a stochastic environment. 

Kussell and Leibler \cite{KL05} model the phenomenon of when bacteria change their phenotype to adjust to the stochastic environment.  The authors assume there are $n$-phenotypes and linear growth  where the fitness and switch to another phenotype is contingent on the current state of the environment. The stochastic environment is modeled by a $k$ atom state continuous time Markov process, and is independent of the evolution of the population.  In a current state, if a certain phenotype's fitness is comparatively small, this increases the probability of phenotype switching. The authors then derive the optimal long-term growth rate.

Markovian switching has also been applied to Lotka-Volterra and epidemiological population dynamics \cite{GGMP12, TDHS06}. Gray et al \cite{GGMP12} assumed a deterministic susceptible-infected-susceptible model and changed parameters according to a continuous Markov chain.  The author discovered that the parameters coupled with the unique invariant measure of the Markov chain gave essentially new rates and found similar inequalities for either an endemic to occur or for the disease to become negligible. Takeuchi et al \cite{TDHS06} analyzed the switching between two deterministic Lotka-Volterra models and showed that this system is neither permanent nor dissipative. 

Fudenberg and  L. A. Imhof \cite{FL12} applied a simpler method where the event that switches the fitness of a population modeled by a Moran process was independent and identically distributed. The authors assumed a two state switched system, considered the mean of the fitness, and compared the switched fitnesses to derive their results. 

Considering both a discrete Moran and a deterministic continuous time replicator dynamic, Harper et al \cite{HFV13} applied a method similar to of Fudenberg and  L. A. Imhof to determine whether the mean game of the switched system was either a strategy 1  dominant (prisoner's dilemma), strategy 2  dominant (prisoner's dilemma), coordination game, or mixed strategy  dominant (hawk-dove). For the continuous time replicator dynamic, the authors determined this classification by comparing the ratios of the difference between the payoffs of the two underlying games and the ratio of whether the event will occur or not. These results are different than the ones derived in this paper.

We analyze a Markovian switched stochastic replicator dynamic with two strategies and determine conditions for this ``new" game to be classified in one of the four games mentioned in the previous paragraph. The times between jumps to another state for the continuous time Markov chain is assumed to have an exponential distribution. Since the switched systems are stochastic, the classifications are similar to the ones given by Fudenberg and Harris \cite{FH92},  in that the inequality of a payoff of pure strategy against itself and the other strategy are perturbed by half the difference of the variances (perturbation from the white noise), and the \textit{comparison} of the transition from a fixed state to the other states (perturbation from the Markov chain). Since the switching indirectly perturbs the dynamic, appropriately determined constants (that are not unique) are associated to a particular state. The difference between the constant of the fixed state and another state, multiplied by this transition rate, compare this transition. The sum of these terms encompasses the entire transition \textit{comparison} for this state. For example, if the dynamic switches between two states where the transition rates are equal, then the addition/subtraction of an appropriately sized constant to the inequalities derived by Fudenberg and Harris \cite{FH92} determine the proper inequalities for the dynamic.

To help create intuition, we first consider a Markov chain in a two atom state space, then extend the analogous results to the general finite atom state space.  To illustrate the conditions for the long-run behavior, we give an example of cooperation in a stochastic environment where defection is punished in one environment, and not punished in the other. The efficacy of punishment is then explored.

\section{Stochastic Replicator Dynamic}
Consider a two-player symmetric game, where $a_{ij}$ is the payoff to a player using pure strategy $S_i$ against an opponent employing strategy $S_j$, and take $A=(a_{ij})$ as the payoff matrix. Within a population we assume that every individual is programmed to play a pure strategy. For $i=1,2$, let $p_i(t)$ be the size of the subpopulation that plays strategy $S_i$ at time $t$, which we denote as the $i^{th}$ subpopulation. Furthermore, define $\mathbf{p}(t):= (p_1(t),  p_2(t) )^T $, $\displaystyle P(t) : = p_1(t) +p_2(t)$, and $\mathbf{s}(t) := (s_1(t), s_2(t))^T$ where $s_i(t) := p_i(t)/P(t)$ \Big( the frequency of the $i^{th}$ subpopulation\Big). When a player in the $i^{th}$ subpopulation is randomly matched with another player from the entire population, $\big(A\mathbf{s}(t) \big)_i$ is the average payoff for this individual, which we take to be the fitness of the player. We assume growth is proportional to fitness: 
\begin{equation*}
 \dot{p}_i(t) = p_i(t) \big( A \mathbf{s}(t) \big)_i,
\end{equation*}  
and hence
\begin{equation*}
\dot{s}_i(t) = s_i(t) \bigg( \big(A \mathbf{s}(t) \big)_i - \mathbf{s}(t)^T A \mathbf{s}(t) \bigg).
\end{equation*}  
This is the \textbf{replicator dynamic}.  For uniformity of notation, we consider the payoff matrix
$
A=
\left(
\begin{array}{cc}
 a &   b   \\
  c &   d   \\  
\end{array}
\right).
$
Simplifying the dynamic above, we have
\begin{equation*}\begin{split}
\dot{s}_1(t) & = s_1(t) s_2 (t) \Big[ \big( a-c \big)s_1(t) +  \big( b-d \big) s_2(t)   \Big]  \\
\dot{s}_2(t) &  = s_1(t) s_2 (t) \Big[ \big( c-a \big)s_1(t) +  \big( d -b \big) s_2(t)  \Big] .
\end{split}\end{equation*}
Since $s_2(t)=1-s_1(t)$ (see \cite{cressman2003evolutionary}), we may focus on the dynamic of $\displaystyle \dot{s}_1(t)  = s_1(t) \big( 1-s_1(t) \big) \Big[ -b + d  + \big( a - c  +  d-b \big) s_1(t)   \Big] $. In this dynamic, if:
\begin{enumerate}
\item $a>c$ and $d<b$ then $S_1$ is the only  dominant strategy (strategy 1 dominant), $(1,0)$ is stable and $(0,1)$ is unstable;
\item $a<c$ and $d>b$ then $S_2$ is the only  dominant strategy (strategy 2 dominant), $(1,0)$ is unstable and $(0,1)$ is stable;
\item $a>c$ and $d>b$ then $S_1$ and $S_2$ are the only dominant strategies where $(1,0)$ and $(0,1)$ are stable, and convergence to either point is contingent on the initial condition (coordination game);
\item and $a<c$ and $d < b$ then $\displaystyle \bigg( \frac{a-c}{a-c+d-b} , \frac{d-b}{a-c+d-c}  \bigg)$ is the only evolutionary stable strategy (mixed strategy  dominant).
\end{enumerate}
In this paper we give conditions that correspond to one of the games listed above. We now describe the stochastic replicator dynamic and discuss the analogous characterizations. 

Throughout this paper, we have the complete probability space $(\Omega, \mathcal{F} , P)$, with the filtration $\{ \mathcal{F}_t \}_{ t\in\mathbb{R}_+ }$, where $\mathcal{F}_0$ contains all of the null sets of $\mathcal{F}$, and the filtration is right-continuous. Fudenberg and Harris \cite{FH92} consider a continuous time stochastic replicator dynamic by first assuming
\begin{equation*}
 dp_i(t) = p_i(t) \bigg( \big(A \mathbf{s}(t) \big)_i dt + \sigma_idW_i(t) \bigg), 
\end{equation*}
for $\sigma_i \in \mathbb{R}_+ $ and $W_i(t)$ a pairwise independent standard Wiener processes. For $i=1,2$ and $j = 3- i$, It\^o's lemma yields 
\begin{equation*}
ds_i(t) = s_i(t) s_j (t) \bigg[ \Big( \big(A \mathbf{s}(t) \big)_i - \big(A \mathbf{s}(t) \big)_j \Big)dt   + \Big(  \sigma_j^2  s_j (t) -  \sigma_i^2 s_i(t) \Big)dt + \Big( \sigma_i dW_i (t)  -  \sigma_j dW_j(t) \Big)  \bigg] .
\end{equation*}
This is known as the \textbf{stochastic replicator dynamic}. The idea behind this model is that randomness comes from the aggregate shock, or population level interactions, that affects the fitness of each type. The only stationary points for this dynamic are the vertices of the simplex. 

 Keeping the payoff matrix as 
$\displaystyle A=\left(
\begin{array}{cc}
 a &  b   \\
 c &  d
\end{array}
\right)
$, we have the dynamic
 \begin{equation*}\begin{split}
ds_1(t) & = s_1(t) s_2 (t) \bigg[ \big( a-c \big)s_1(t) +  \big( b-d \big) s_2(t)   + \sigma_2^2  s_2(t) -  \sigma_1^2 s_1(t)  \bigg] dt + s_1(t) s_2 (t)\sigma dW (t) \\
ds_2(t) &  = s_1(t) s_2 (t) \bigg[ \big( c-a \big)s_1(t) +  \big( d -b \big) s_2(t)   + \sigma_1^2  s_1(t) -  \sigma_1^2 s_2(t)  \bigg] dt - s_1(t) s_2 (t)\sigma dW (t),
\end{split}\end{equation*}
where $W(t)$ is a standard Brownian motion, and $\displaystyle \sigma = \sqrt{ \sigma_1^2 + \sigma_2^2  }$.

\bigskip
 \begin{re}\label{re1}
For $0<x<1$, $\displaystyle P_x \big( s(t) \in (0,1) \big)=1$ for all finite $t$. See \cite{I05} for further information.
\end{re}

Since $s_2(t)=1-s_1(t)$,  we may focus on the dynamic
\begin{equation}\begin{split} \label{fh}
ds_1(t) & = s_1(t) \big( 1-s_1(t) \big) \bigg[ \big( -b + d +  \sigma_2^2 \big)    + \big( a - c  -  \sigma_1^2 +   d-b -\sigma_2^2 \big)s_1(t)   \bigg] dt  + s_1(t) \big( 1-s_1(t) \big) \sigma dW (t).
\end{split}\end{equation}
For simplicity we write $s(t)$ instead of $s_1(t)$. Fudenberg and Harris \cite{FH92} assumed the same inequalities of the payoffs as above and derived the following proposition to determine the conditions for the stability of the process.

\bigskip
\begin{prop}[Fudenberg and Harris \cite{FH92}] \label{fh_r}
For the dynamic given by  Equation \eqref{fh}, and initial condition $0<x_0<1$, we have that
\begin{enumerate} 
\item (Strategy 1  Dominant) If $a -c > \big( \sigma^2_1 - \sigma^2_2  \big)/2$ and $d-b <  \big( \sigma^2_2 - \sigma^2_1  \big)/2$ then $s(t) \to 1$ as $t \to \infty$ a.s.
\item  (Strategy 2  Dominant) If $a -c <  \big( \sigma^2_1 - \sigma^2_2  \big)/2$ and $d-b >  \big( \sigma^2_2 - \sigma^2_1  \big)/2$, we have $s(t) \to 0$ as $t \to \infty$ a.s.
\item (Coordination) If $a -c  >  \big( \sigma^2_1 - \sigma^2_2  \big)/2$ and $d-b >  \big( \sigma^2_2 - \sigma^2_1  \big)/2$ then $s(t) \to 1$ as $t \to \infty$ with probability $\frac{I_1(x_0)}{I_1(x_0) + I_2(x_0) }$ and $s(t) \to 0$ as $t \to \infty$ with probability $\frac{I_2(x_0) }{I_1(x_0) + I_2(x_0) }$. (For the exact values of $I_1(x_0)$ and $I_2(x_0)$, see \cite{FH92})
\item (Mixed Strategy Dominant) If $a -c  <  \big( \sigma^2_1 - \sigma^2_2  \big)/2$ and $d-b <  \big( \sigma^2_2 - \sigma^2_1  \big)/2$ then $\displaystyle P_{x_0}\Big(\liminf_{t \to \infty} s(t) =0 \Big) = P_{x_0}\Big( \limsup_{t \to \infty} s(t) =1 \Big) =1$. In fact, the process is positive recurrent with a unique invariant measure.
\end{enumerate}
\end{prop}

\bigskip
In this model there is a possibility that the variances will trump an evolutionary stable strategy, e.g., $\displaystyle c<a<c+\frac{\sigma^2_1 - \sigma^2_2 }{2}$, or change a dominated pure strategy dominant, e.g., $\displaystyle b > d > b+\frac{\sigma^2_2 - \sigma^2_1 }{2}$.

Throughout this paper, the ``stochastic'' classifications in Proposition \ref{fh_r} will be the reasoning for a vertex being either stable or unstable. We now ready to define the switched dynamic.

\section{Markovian Switching}
Let $r(t)$ be a continuous time Markov chain with the state space $\{ 1,2\}$ and generator
$
Q=
\left(
\begin{array}{cc}
 -q_{12} & q_{12}     \\
 q_{21} & - q_{21} 
\end{array}
\right),
$
where $q_{12}>0$ is a transition rate from state 1 to state 2, and $q_{21}>0$ is a transition rate from state 2 to state 1. So for $\delta>0$, $P\Big( r(t+\delta)=2 \Big| r(t)=1 \Big)= q_{12} \delta + o(\delta)$ and $P\Big( r(t+\delta)=1 \Big| r(t)=2 \Big)= q_{21} \delta + o(\delta)$.  Given an increasing sequence of times that the process jumps states (note these are stopping times), say $0=\tau_0<\tau_1< \ldots <\tau_k \to \infty$, r(t) may be written as
$$
r(t) =\sum_{ k=0}^{\infty} r( \tau_k) I_{ [\tau_k, \tau_{k+1}) } (t).
$$
Given that $ r( \tau_k)=i$, the times between jumps are exponentially distributed with parameter $q_{ij}$, with $j=3- i$. Hence, for any $T \geq 0$, $P\Big( \tau_{k+1} - \tau_{k} \geq T \Big| r(\tau_{k})=1 \Big) = e^{ -q_{12} T}$, and similarly for the initial condition $r(\tau_{k})=2$. Finally, this Markov chain has a unique stationary distribution $\Pi = \big( \pi_1, \pi_2\big)$, where $\displaystyle \pi_1 = \frac{q_{21} }{ q_{12} +q_{21} }$ and $\displaystyle \pi_2 = \frac{q_{12} }{ q_{12} +q_{21} }$. For further information, see Anderson \cite{A91}. 

For two stochastic replicator dynamics, define the payoff matrices $\displaystyle A_1=\left(
\begin{array}{cc}
 a_1 &  b_1   \\
 c_1 &  d_1
\end{array}
\right)
$,
and 
$\displaystyle A_2=\left(
\begin{array}{cc}
 a_2 &  b_2   \\
 c_2 &  d_2
\end{array}
\right),
$
and the corresponding diffusion coefficients $\sigma_{1 1}$ and $\sigma_{2 1}$, and $\sigma_{1 2}$ and $\sigma_{2 2}$. Assuming that $r(t)$ is independent of the Brownian perturbation, we consider the replicator dynamic with Markovian switching given by
 \begin{equation}\begin{split}  \label{M}
dY\big(s(t), r(t) \big) & := s(t) \big( 1-s(t) \big) \bigg[ \big( -d_{r(t)} + b_{r(t)} +  \sigma_{2 r(t)}^2 \Big)    + \big( a_{r(t)} - c_{r(t)}  -  \sigma_{1r(t)}^2 +   d_{r(t)} -b_{r(t)} -\sigma_{2 r(t)}^2 \big)s(t)   \bigg] dt \\
& + s(t) \big( 1-s(t) \big) \sqrt{ \sigma_{1 r(t)}^2 + \sigma_{2 r(t)}^2 } dW (t) \\
& := s(t) \big( 1-s(t) \big) \bigg[  -B_{r(t)}  + \big( A_{r(t)}  +   B_{r(t)} \big)s(t)   \bigg] dt + s(t) \big( 1-s(t) \big) \sigma_{r(t)} dW (t),
\end{split}\end{equation}
where $\displaystyle B_{r(t)} :=d_{r(t)} -b_{r(t)} -\sigma_{2 r(t)}^2$, $\displaystyle A_{r(t)} := a_{r(t)} - c_{r(t)}  -  \sigma_{1r(t)}^2$, and $\sigma_{r(t)}:=\sqrt{ \sigma_{1 r(t)}^2 + \sigma_{2 r(t)}^2 }$. To simplify the notation, we set $S(t) =Y\big(s(t), r(t) \big)$.

For $L$ the infinitesimal operator, and $V(x,i) \in [0,1] \times \{1,2\}$,  where for $i=1,2$, and $j:=3-i$, $V( \cdot ,i)$ is twice continuously differentiable, we have 
 \begin{equation*}\begin{split}
LV(x,i) & = x \big( 1-x \big) \Big[ -B_{i}  + \big( A_{i}  +   B_{i} \big) x  \Big] V'(x,i) + \frac{ \sigma_i^2}{2} x^2(1-x)^2 V''(x,i) \\
&  +  q_{ij} V(x,j) - q_{ij} V(x,i).
\end{split}\end{equation*}
For more information, see page 48 in \cite{MY06}, and page 103 in \cite{S89}.

We now define the notation of stochastic stability that will be used throughout this paper. Although very similar to the definitions given in Khasminskii \cite{RK80}, they are very natural to the switched dynamic. We follow \cite{KZY07, MY06}, and note that the definition is given for the point $x=0$ and for a process evolving on the unit interval, but has a natural extension to $x=1$. Although we first consider the case when $r(t)$ is in a two atom state, we give the definition for a  general finite $n$ atom state space, which we call $I:=\{1,2,\dots,n\}$.

 We define $P_{x,i}$ as the probability measure corresponding to $S(t)$ when $s(0)=x$ and $r(0)=i$, almost surely. Throughout this paper, the initial conditions are assumed to almost surely hold.  
 
 \bigskip
\begin{de}
The stationary point $x=0$ is said to be:
\begin{enumerate}[1.]
\item \textbf{stable in probability} if, for any $\epsilon>0$ and $i\in I$,
$$
\lim_{x\to 0} P_{x,i} \bigg( \sup_{t \geq 0} S(t)> \epsilon \bigg)=0;
$$
\item \textbf{asymptotically stable in probability} if it is stable in probability and for any $i\in I$,
$$
\lim_{x\to 0} P_{x,i} \bigg( \lim_{t\to\infty} S (t)=0\bigg)=1
$$
\end{enumerate}
\end{de}
 
\bigskip
We now define the possible properties that the dynamic will hold. The definitions are given for a process evolving in the unit interval, and are the adjusted definitions for a process on the real line. Please see \cite{MY06} for the general definitions.

\bigskip
\begin{de}  
For each of the following definitions, take the pair $(x,i)$ as the initial condition. 
\begin{enumerate}[1.]
\item A Markov process $S(t)$ is said to be regular if for any finite time $T>0$,
$$
P_{x,i} \bigg( \sup_{0 \leq t \leq T} S(t) = 0 \ \textnormal{or} \ 1 \bigg)=0
$$

\item For $\displaystyle U:= D \times J$, where $D \subsetneq [0,1]$ nonempty and $J \subset I$, and $\displaystyle \tau_U : = \inf\Big\{ t \geq 0 : S(t) \in U \Big\}$, a Markov process $S(t)$ is called recurrent with respect to $U$ if it is regular and for any finite time $T>0$,
$$
P_{x,i} \Big( \tau_{U} < \infty \Big)=1,
$$
where $(x,i)$ is an arbitrary initial condition in $ D^c \times I$. (The notation $D^c$ means the complement of the set $D$.)

\item If $S(t)$ is not recurrent, it is called transient. 

\item The process is called positive recurrent with respect to the set $U$ if it is recurrent and $D\subset (\omega_0,\omega_1)$, where $0<\omega_0<\omega_1<1$.

\end{enumerate}
\end{de}

\bigskip
In this paper we apply the stochastic Lyapunov method to derive the stability or instability of each vertex, and characterize long-run behavior of the dynamic by combining the stability or instability property for each vertex and utilizing these properties. Similar to the deterministic Lyapunov method, the stochastic Lyapunov method transforms the process so that it is a positive valued supermartingale (defined below), which implies that the process is decreasing.

\bigskip
\begin{de}
For $0 \leq s \leq t$ and a Markov process $X$, define $\displaystyle E\big[ X(t) \big| \sigma\{ X(s_0) : s_0 \leq s\}  \big]$ as the conditional expectation of the process at time $t$ dependent on the history of the process up to time $s$, (where $\sigma\{ X(s_0) : s_0 \leq s\}$ is the $\sigma$-algebra of the process up to time $s$). We call $X$ \textbf{martingale} if $\displaystyle E\big[ X(t) \big| \sigma\{ X(s_0) : s_0 \leq s\}  \big]= X(s)$, a \textbf{submartingale} if $\displaystyle E\big[ X(t) \big| \sigma\{ X(s_0) : s_0 \leq s\} \big] \geq X(s)$, and a \textbf{supermartingale} if $\displaystyle E\big[ X(t) \big| \sigma\{ X(s_0) : s_0 \leq s\}  \big] \leq X(s)$.
\end{de}

\section{Analysis of the Model}
Before we begin our analysis, we show an important characteristic of Equation \eqref{M}. Given any strict subinterval of the unit interval, $S(t)$ leaves this interval in finite time. Since each process of the switched dynamic has this property, the result is very natural. However, one may wonder if the switching could keep a significant number of sample paths within this interval. When reading the proof, notice that it does not depend on the size of the state space of $r(t)$, but that the state space has a finite number of atoms.

\bigskip
\begin{prop}\label{leave}
For $(x_0,x_1)\subset (\omega_0,\omega_1)$, where $0<\omega_0<\omega_1<1$, $x \in(x_0,x_1)$, and $\tau$ defined to be the time $S(t)$ leaves this interval, we have $ E_{x,i} \big[\tau \big] <\infty.$
\end{prop}
\begin{proof}
Define the positive function $V(x,i)=e^{\gamma} - e^{ \gamma x}$, where $\gamma>0$. Notice that $V(x,i)$ does not depend on $i$. For any $i$, we see that
\begin{equation*}\begin{split}
LV(x,i) & = - \gamma e^{ \gamma x} x (1-x ) \big[ -B_i + \big( A_i + B_i \big)x  \big]  - \gamma^2 e^{ \gamma x} \frac{ \sigma_i^2}{2}  x^2 (1-x)^2 \\
& = - \gamma e^{ \gamma x} x (1-x )  \Big\{ -B_i + \big( A_i + B_i \big)x   +  \gamma \frac{ \sigma_i^2}{2}  x (1-x) \Big\}.
\end{split}\end{equation*}
We now choose $\gamma$ large enough so that $\displaystyle  -B_i + \big( A_i + B_i \big)x   +   \gamma \frac{ \sigma_i^2}{2}  x (1-x) > 0$ for all $x \in ( x_0, x_1 )$, and $i \in \{1,2\}$. Thus, there exists a constant $K>0$, such that $L V(x,i) \leq -K$ for any $i$. Dynkin's formula yields
$$
0 \leq E_{x,i} \Big[  V\big( \tau \wedge t \big) \Big] = V(x,i) + E_{x,i}   \Big[ \int^{ \tau \wedge t}_{0} LV\big( S(u)  \big) du  \Big] \leq V(x,i) - K E_{x,i}  \big[ \tau \wedge t  \big],
$$
which implies that $\displaystyle E_{x,i}   \big[ \tau \wedge t  \big] \leq V(x,i)/K$. Taking $t \to \infty$, the monotone convergence theorem tells us that that $E_{x,i}  \big[  \tau  \big]<\infty$.
\end{proof}

Borrowing and adjusting the Lyapunov function defined in \cite{KZY07, YZ07}, we give conditions for 0 and 1 to be stochastically stable or unstable. Take $0<\alpha<1$, and constants $c_1$ and $c_2$. These constants are not unique and will be taken accordingly in order establish stability or instability of a particular vertex. Define the four positive Lyapunov functions as $V_0^{\pm }(x,i)=\big(1\mp\alpha c_i \big) x^{\pm\alpha} $, and $V_1^{\pm}(x,i)=\big(1\pm\alpha c_i \big) \big(1-x \big)^{\pm\alpha} $. Note that for each $i$, $\displaystyle \lim_{ x \searrow 0} V_0^{+}(x,i)=0=\lim_{ x \nearrow 1} V_1^{+}(x,i)$, and  $\displaystyle \lim_{ x \searrow 0} V_0^{-}(x,i)=\infty=\lim_{ x \nearrow 1} V_1^{-}(x,i)$. These limits in conjunction with the conditions for the dynamic to be a supermartingale will tell us whether our process is stable or unstable for the respective vertex. Conditions for the dynamic to be a supermartingale are given below. 

Applying the infinitesimal generator to each function, for $i=1,2$, and take $j:=3-i$, we have
$$
LV_0^{+}(x,i) = \alpha \big( 1- \alpha c_i \big) x^{\alpha}\Bigg\{  \big(1-x\big) \Big[ -B_i + ( A_i+B_i)x \Big]  + \frac{ \alpha -1}{2} \sigma_i^2 (1-x)^2 + q_{ij} \frac{ c_i -c_j }{ 1- \alpha c_i   }  \Bigg\},
$$
$$
LV_0^{-}(x,i) =  -\alpha \big( 1+ \alpha c_i \big) x^{-\alpha}\Bigg\{  \big(1-x \big) \Big[ -B_i + ( A_i+B_i)x \Big] + \frac{ -\alpha -1}{2} \sigma_i^2 (1-x)^2 + q_{ij} \frac{ c_i -c_j }{ 1+ \alpha c_i   }  \Bigg\},
$$
$$
LV_1^{+}(x,i) =  -\alpha \big( 1+ \alpha c_i \big) \big( 1-x \big)^{\alpha}\Bigg\{  x \Big[ -B_i + ( A_i+B_i)x \Big] + \frac{ -\alpha + 1}{2} \sigma_i^2 x^2 + q_{ij} \frac{ c_i -c_j }{ 1+ \alpha c_i   }  \Bigg\},
$$
and
$$
LV_1^{-}(x,i) =  \alpha \big( 1- \alpha c_i \big) \big( 1-x \big)^{-\alpha}\Bigg\{  x\Big[ -B_i + ( A_i+B_i)x \Big] + \frac{ \alpha +1}{2} \sigma_i^2 x^2 + q_{ij} \frac{ c_i -c_j }{ 1+ \alpha c_i   }  \Bigg\}.
$$
 If there exist a neighborhood (within the simplex) around 0 or 1 such that $LV_0^{\pm}(x,i) \leq 0$ or $LV_1^{\pm}(x,i) \leq 0$ for $x$ in the respective neighborhoods, then $V\big( S(t) \big)$ is a supermartingale in these neighborhoods. To accomplish this goal, we find conditions for the functions defined in the curly brackets above to be the appropriate sign at either $x=0$ or $x=1$. Plugging in $x=0$ or $x=1$ as appropriate yields the following assumptions:

\begin{ass} \label{ass}
Assume there exists $0<\alpha<1$, $c_1$, and $c_2$, where one of the following inequalities holds for each $i=1,2$, and $j :=3- i$:
\begin{enumerate}[(i)]
\item $1- \alpha c_i >0$ and $-B_i +\frac{\alpha -1}{2} \sigma_i^2 + q_{ij} \frac{ c_i -c_j  }{ 1- \alpha c_i } < 0$;
\item $1+ \alpha c_i >0$ and $-B_i +\frac{-\alpha -1}{2} \sigma_i^2 + q_{ij} \frac{ c_i -c_j  }{ 1+ \alpha c_i } > 0$;
\item $1+\alpha c_i >0$ and  $A_i +\frac{-\alpha + 1}{2} \sigma_i^2 + q_{ij} \frac{ c_i -c_j  }{ 1+ \alpha c_i } > 0$;
\item $1- \alpha c_i >0$ and  $A_i +\frac{\alpha + 1}{2} \sigma_i^2 + q_{ij} \frac{ c_i -c_j  }{ 1- \alpha c_i } < 0$.
\end{enumerate}
\end{ass}

\bigskip
Notice that these assumptions are an extension to the inequalities derived by Fudenberg and Harris \cite{FH92}. Furthermore, note that the term $\displaystyle \frac{|\alpha|}{2} \sigma_i^2$ is negligible.

In order  to utilize the assumptions to prove stability/instability properties of the switched replicator dynamic for each vertex, for any sufficiently small $\epsilon$ and $\delta$, any initial condition $\epsilon<x< \delta$ and $i\in\{1,2\}$, or $1-\delta< x < 1-\epsilon$ and $i\in\{1,2\}$, the exit time out of the interval $(\epsilon,\delta)$, or $(1-\delta,1-\epsilon)$, has to be almost surely finite. Proposition \ref{leave} gives us this characteristic.

We are now ready to state a theorem about the stability or instability of the vertices. The argument is very similar to the proofs in Theorems 5.3.1, 5.4.1, and 5.4.2 in \cite{RK80}. Since the proofs of these theorems only utilize the assumptions of the Markov process being continuous and the ability to apply the strong Markov property, one may see that there is natural extension to our theorem. For brevity, we omit the proof.

\bigskip
\begin{thm}\label{thm_main}
 For the dynamic defined by Equation \eqref{M}, if: 
\begin{enumerate}[(1)]
\item Assumption \ref{ass}(i) holds, then $x=0$ is asymptotically stable in probability;
\item Assumption \ref{ass}(ii) holds, then $x=0$ is unstable in probability;
\item Assumption \ref{ass}(iii) holds, then $x=1$ is asymptotically stable in probability;
\item Assumption \ref{ass}(iv) holds, then $x=1$ is unstable in probability.
\end{enumerate}
\end{thm}

\bigskip
Taking the conditions in Theorem \ref{thm_main}  we establish the long-run behavior of the switched stochastic replicator dynamic. For example, if Assumption \ref{ass}(i) (the point 0 is  asymptotically stable) and  Assumption \ref{ass}(iv) (the point 1 is unstable) hold then the our dynamic will converge to 0 almost surely, which is the strategy 2  dominant game. This is Proposition \ref{0as}. However, for the transient/recurrent properties to hold, the dynamic needs to be regular. One may see that this property does holds by Remark \ref{re1}.

Since the process evolves on the unit interval, if the dynamics that are switched all have the same unstable vertex, then this vertex will also be unstable for the switched dynamic. This can be shown by setting $c_1=c_2$ for the appropriate $V_{\cdot}^{-}$ function. The same reasoning my be applied to a stable vertex.

For the proofs of the following propositions, we define the stopping times: $\tau_{\epsilon}$ as the first time the process leaves the interval $(0,\epsilon)$; and $\tilde{\tau}_{\epsilon}$ as the first time the process leaves the interval $(1-\epsilon,1)$. Furthermore, we assume that the initial conditions lie in $(0,1) \times \{1,2\}$.

\bigskip
\begin{prop}[Strategy 2 Dominant]\label{0as}
 For the dynamic defined by Equation \eqref{M}, if Assumptions \ref{ass}(i) and \ref{ass}(iv) hold, then for any initial condition, $S(t)$ converges to 0 almost surely.  
\end{prop}
\begin{proof}
By Assumption \ref{ass}(iv), there exists a $\delta>0$ such that, for $x\in (1-\delta, 1)$ and an $\alpha>0$, $LV(x,i) \leq -\alpha$, for $i\in\{1,2\}$. This tells us that $E_{x,i}[ \tilde{\tau}_{\delta}]<\infty$. Also, by Assumption \ref{ass}(i), there exists an $\epsilon>0$ such that for $x\in(0,\epsilon)$ and $i\in\{1,2\}$, we have $\displaystyle P_{x,i}\Big( \lim_{t\to\infty} S(t)=0 \Big)>0$.

We now show that the process is transient, which implies that $S(t)$ converges to 0 almost surely.  To accomplish this, we apply Theorem 3 in Myen and Tweedie \cite{MT93_1}. Take a constant $\epsilon_0$, where $0<\epsilon_0< \epsilon$, and for the Lebesgue measure $M$ and $A \in \mathfrak{B}\Big( (0,1) \Big)$, define $\Phi\big( A \big) = M\Big( A \cap (\epsilon_0,1-\delta) \Big)$. We have that the measure $\Phi$ is an irreducible measure (see  \cite{DMT95, MT93,MT93_1} for further information). For $x\in (\epsilon_0, \epsilon)$, notice that $P_{x,i}\big( \tau_{\epsilon}=\infty \big)>0$.  Therefore, the assumptions of Theorem 3 hold, and hence our dynamic is transient.
\end{proof}

\bigskip
Notice that the above proof utilizes the characteristic of the process near each stationary point to show almost sure convergence to the point 0. For the case when 0 is unstable and 1 is asymptotically stable, the argument to show that the dynamic converges to 1 almost sure is identical. This gives us the following corollary.

\bigskip
\begin{cor}[Strategy 1 Dominant]
If Assumptions \ref{ass}(ii) and \ref{ass}(iii) hold, then for any initial condition, S(t) converges to 1 almost surely. 
\end{cor}

\bigskip
The next proposition shows that when both stationary points are asymptotically stable in probability, the process is transient and will converge to either endpoint almost surely. Due to the complexity of the process, given an initial condition, we are unable to derive probabilities of converging to 0 and converging to 1. When there is no switching between processes, the probabilities are explicit and given in Fudenberg and Harris \cite{FH92}.

\bigskip
\begin{prop}[Coordination Game]
 For the dynamic defined by Equation \eqref{M}, if Assumptions \ref{ass}(i) and \ref{ass}(iii) hold, then S(t) is transient and converges to 0 or 1 almost surely.  
\end{prop}
\begin{proof}
Take an arbitrarily small $\epsilon>0$. By Assumptions \ref{ass}(i) and \ref{ass}(iii), there exists a $\delta>0$ such that $\displaystyle  P_{x,i}  \Big( \lim_{t\to\infty} S(t)=0  \Big) \geq 1- \epsilon$ for $x < \delta$, or  $\displaystyle  P_{x,i}  \Big( \lim_{t\to\infty} S(t)=1 \Big) \geq 1- \epsilon$ for $x >1-\delta$. Define $\displaystyle F= \Big\{ \lim_{t\to\infty} S(t)=1 \ \mbox{or} \ 0 \Big\}$, and $\hat{\tau}_{\delta}$ as the first time the process leaves the interval $(\delta, 1-\delta)$.

\medskip
Proposition \ref{leave} tells us that $\displaystyle  \hat{\tau}_{\delta}$ has finite mean. Therefore, the strong Markov property yields $\displaystyle  P_{x,i}  \big(F \big) =  E_{x,i}  \Big[ E_{S (\hat{\tau}_{\epsilon} ) } [ I_{F} ] \Big] \geq 1 - \epsilon.$ Since $\epsilon$ was arbitrary, we are able to conclude the statement.
\end{proof}

\bigskip
The proposition below gives conditions for the dynamic to be positive recurrent. For example, if the dynamic switches between two ``stochastic" mixed strategy dominant games, taking $c_1=c_2$ for  $V_0^-$ and $V_1^-$ tells us that the condition holds. 

\bigskip
\begin{prop}[Mixed Strategy Dominant]
 For the dynamic defined by Equation \eqref{M}, if Assumptions \ref{ass}(ii) and \ref{ass}(iv) hold, then S(t) is positive recurrent.  
\end{prop}
\begin{proof}
By Assumption \ref{ass}(iv), there exists a $\delta>0$  and $\alpha_1>0$ such that for any $x\in (1-\delta, 1)$ and $i\in\{1,2\}$, we have $LV(x,i)\leq -\alpha_1$. Thus $E_{x,i}[ \tilde{\tau}_{\delta}]<\infty$.  Moreover, Assumption \ref{ass}(ii) says there exists an $\epsilon>0$ and $\alpha_0>0$ such that for $x\in(0,\epsilon)$ and $i\in\{1,2\}$, we have $LV(x,i)\leq -\alpha_0$. Thus $E_{x,i}[ \tau_{\epsilon}]<\infty$. Therefore, the process hits the set $(\epsilon, 1-\delta)$ in finite time. The strong Markov property tells us that the process is positive recurrent in $(\epsilon, 1-\delta)$.  Proposition \ref{leave} yields that the process is positive recurrent for any strict subinterval of $(0,1)$.
\end{proof}

\bigskip
\begin{re}
When $q_{12} = q_{21}$, whichever process has the ``stronger'' stability or instability to the particular vertex will dictate the dynamic. To illustrate this, consider the case when $\displaystyle - B_1 + \frac{ -\alpha - 1}{2} \sigma_1^2>0$ and $\displaystyle - B_2 + \frac{ -\alpha - 1}{2} \sigma_2^2 < 0$, where $\alpha\ll 1$. Since $q_{12} = q_{21}$, to show that Assumption \ref{ass}(ii) may hold, we will subtract a positive constant from the first term, and add the same constant to the second term. If $\displaystyle - B_1 + \frac{ - \alpha - 1}{2} \sigma_1^2> \Big| - B_2 + \frac{ - \alpha - 1}{2} \sigma_2^2\Big|$ (which means the process for $r(t)=1$ has a ``stronger'' instability to $x=0$ than the process $r(t)=2$ has stability at $x=0$), then we may find a $k>0$ where   $\displaystyle - B_1 + \frac{ \alpha - 1}{2} \sigma_1^2 - k>0$ and $\displaystyle - B_2 + \frac{ \alpha - 1}{2} \sigma_2^2 +k > 0$, which tells us that $S(t)$ is unstable at $x=0$. 
\end{re}

\section{Example of Defectors}
Taylor et al \cite{TCI07} analyzed a prisoner's dilemma game between cooperators $C$ and defectors $D$, where defection was punished. Assuming a payoff $b>0$ for cooperation and a cost $c>0$, where $b>c$, the payoffs are
$
\begin{array}{c|cc}
 &C  & D\\ \hline
C & b \sm c & \sm c \\
D & b & 0   \\
\end{array}
$,
leaving the defection strategy the only dominant strategy. When punishment is taken into account, the payoffs are
$
\begin{array}{c|cc}
& C  &  D \\ \hline
C & b \sm c & \sm c \\ 
D & b \sm \frac{\gamma}{2}(b+c) & 0   \\
\end{array}
$,
where for a large enough $\gamma$, this is a coordination game. We consider the scenario where punishment for defection follows a stochastic environment.  Take the payoff matrices for the two environment as
$
\left(
\begin{array}{cc}
 1 & \sm 1   \\
2  &  0    \\  
\end{array}
\right)
$
and
$
\left(
\begin{array}{cc}
 1 &  \sm 1  \\
  .2 &  0    \\  
\end{array}
\right)
$,
and the variances values are $\sigma^2_{21}=.1=\sigma^2_{22}$, $\sigma^2_{11}=.15=\sigma^2_{12}$, and $\sigma^2_{1}=.5=\sigma^2_{2}$. We call the game without punishment state 1 and the game with punishment state 2. Since in both states the point $(0,1)$ (the population defecting) is stable, this is also true in the switched environment. This is Assumption \ref{ass}(i), which written out is
\begin{equation*}\begin{split}
\sm .9 & + .5\big( \alpha -1 \big)  + q_{12} \frac{ c_1 -c_2  }{1 - \alpha c_1} < 0  \\
\sm .9 & + .5 \big( \alpha -1 \big) + q_{21} \frac{ c_2 -c_1  }{1 - \alpha c_2} < 0  \\
\end{split}\end{equation*}  
where $0<\alpha<1$, $1 - \alpha c_1>0$, and $1 - \alpha c_2>0$. Taking $c_1=c_2$ gives us the inequalities. For the vertex $(1,0)$ (population of cooperators) to be stable, one would assume that the transition rate from state 2 to state 1 is rather small. Hence, if this transition rate is large, $(1,0)$ would be unstable. The inequalities needed for $(1,0)$ to be unstable (Assumption \ref{ass}(vi)) are
\begin{equation*}\begin{split}
\sm 1.15 & + .5 \big( \alpha +1 \big)  + q_{12} \frac{ c_1 -c_2  }{1 - \alpha c_1} < 0  \\
.65 & + .5 \big( \alpha +1 \big) + q_{21} \frac{ c_2 -c_1  }{1 - \alpha c_2} < 0  \\
\end{split}\end{equation*}  
where, again, $0<\alpha<1$, $1 - \alpha c_1>0$, and $1 - \alpha c_2>0$. Thus, if $q_{21} > q_{12}$ one may be able to find $\alpha$, $c_1$, and $c_2$ to make these inequalities hold. However, if $q_{12} > q_{21}$ then it might be possible for $(1,0)$ to be stable, i.e.,
\begin{equation*}\begin{split}
\sm 1.15 & + .5 \big( \sm \alpha +1 \big)  + q_{12} \frac{ c_1 -c_2  }{1 +\alpha c_1} > 0  \\
.65 & + .5 \big( \sm \alpha +1 \big) + q_{21} \frac{ c_2 -c_1  }{1 + \alpha c_2} > 0  \\
\end{split}\end{equation*}  
which is Assumption \ref{ass}(iii). In particular, take $q_{12}=1$ and $q_{21}=3$. Clearly, we should be able to find constants such that the inequalities in Assumption \ref{ass}(vi) hold. In fact, define $\alpha=.001$, $c_1=-1$, and $c_2=-1.5$. Then
\begin{equation*}\begin{split}
\sm 1.15 & + .5 \big( \alpha +1 \big)  + 1\frac{ c_1 -c_2  }{1 - \alpha c_1} =-.9 < 0  \\
.65 & + .5 \big( \alpha +1 \big) + 3\frac{ c_2 -c_1  }{1 - \alpha c_2} \approx-6.3382 < 0 . \\
\end{split}\end{equation*}

\section{Analysis With a Markov Chain with General Finite State Space. }
In this section we consider a generalized version of Equation \eqref{M}, where $r(t) \in \{1,2, \ldots,n\}$, for some finite integer $n$. The analysis for the case when the state space was $\{1,2 \}$ used only the characteristic of the dynamic near the vertices, and not the assumption that $r(t)$ jumps between two atoms. Therefore, we have analogous results with identical proofs. However, since we have a general finite $n$ state space, we need to slightly change the assumptions for stability and instability of the vertices.

\smallskip
Take $r(t)$ to be a continuous time Markov chain in state space $\{1,2, \ldots,n\}$, for some finite integer $n$, with infinitesimal generator $\displaystyle Q=
\left(
\begin{array}{cccc}
  q_{11} & q_{12}  & \ldots & q_{1n}  \\
   \vdots  &  \vdots  & \ddots   &    \vdots        \\                     
 q_{n1}  &  q_{n2}   & \ldots & q_{nn}   
\end{array}
\right).
$  Recall that, for $i\neq j$, $q_{ij}>0$ is the transition rate from $i$ to $j$, and $\displaystyle q_{ii}=-\sum_{j \neq i} q_{ij}$. Furthermore for $j \neq i$, $\displaystyle P\Big( r\big(t+\delta)=j \Big| r(t)=i \Big)=q_{ij}\delta + o(\delta)$, and $\displaystyle P\Big( \tau_{k+1} -\tau_{k} \geq T \Big| r(\tau_k)=i \Big)=e^{q_{ii}T}$ for all $T \geq 0$. Just as in the previous sections, we assume that the Markov chain is independent of the Brownian motion.  

We consider the dynamic
 \begin{equation}\begin{split}\label{rms_n}
dS(t) & = s(t) \big( 1-s(t) \big) \bigg[  -B_{r(t)}  + \big( A_{r(t)}  +   B_{r(t)} \big)s(t)   \bigg] dt + s(t) \big( 1-s(t) \big) \sigma_{r(t)} dW (t),
\end{split}\end{equation}
where $r(t) \in \{1,2, \ldots,n\}$, and the payoff matrix 
$\displaystyle 
A_i=\left(
\begin{array}{cc}
a_i  & b_i    \\
c_i  & d_i     
\end{array}
\right)
$
and variances $\sigma_{1i}$, $\sigma_{2i}$, and $\sigma_i:=\sqrt{ \sigma_{1i}^2+ \sigma_{2i}^2 }$ correspond to the state $r(t)=i$.

For $L_n$ the infinitesimal operator, and $V(x,i) \in [0,1] \times \{1,2, \ldots, n\}$,  where for each $i$, $V( \cdot ,i)$ is twice continuously differentiable, we have 
 \begin{equation*}\begin{split}
 L_nV(x,i) & = x \big( 1-x \big) \Big[ -B_{i}  + \big( A_{i}  +   B_{i} \big) x  \Big] V'(x,i) + \frac{ \sigma_i^2}{2} x^2(1-x)^2 V''(x,i) + \sum_{j=1}^{n} q_{ij} V(x,j).
\end{split}\end{equation*}

\bigskip
Applying $L_n$ to the same Lyapunov functions as in the previous sections, $V_0^{\pm}(x,i)$ and $V_1^{\pm}(x,i)$, we derive the following assumptions.  

\bigskip
\begin{ass} \label{ass_n}
Assume there exists $0<\alpha<1$, $c_1$,  $c_2$, $\ldots$, $c_{n-1}$, and $c_n$ where one of the following inequalities holds for each $i\in \{1,2, \ldots,n\}$:
\begin{enumerate}[(i)]
\item $1- \alpha c_i >0$ and $ \displaystyle -B_i +\frac{\alpha -1}{2} \sigma_i^2 + \sum_{i=1}^{n} q_{ij} \frac{ c_i -c_j  }{ 1- \alpha c_i } < 0$;
\item $1+ \alpha c_i >0$ and $ \displaystyle -B_i +\frac{-\alpha -1}{2} \sigma_i^2 + \sum_{i=1}^{n} q_{ij} \frac{ c_i -c_j  }{ 1+ \alpha c_i } > 0$;
\item $1+\alpha c_i >0$ and  $ \displaystyle A_i +\frac{-\alpha + 1}{2} \sigma_i^2 + \sum_{i=1}^{n} q_{ij} \frac{ c_i -c_j  }{ 1+ \alpha c_i } > 0$;
\item $1- \alpha c_i >0$ and  $  \displaystyle A_i +\frac{\alpha + 1}{2} \sigma_i^2 + \sum_{i=1}^{n} q_{ij} \frac{ c_i -c_j  }{ 1- \alpha c_i } < 0$.
\end{enumerate}
\end{ass}

\bigskip
We use these assumptions to tell us the stability or instability of the process near the stationary points. Again, the arguments follow closely the arguments given in Theorems 5.3.1, 5.4.1, and 5.4.2 in \cite{RK80}, and so we omit the proof.  

\bigskip
\begin{thm}
For the dynamic defined by Equation \eqref{rms_n}, if: 
\begin{enumerate}[(1)]
\item Assumption \ref{ass_n}(i) holds, then $0$ is asymptotically stable in probability;
\item Assumption \ref{ass_n}(ii) holds, then $0$ is unstable in probability;
\item Assumption \ref{ass_n}(iii) holds, then $1$ is asymptotically stable in probability;
\item Assumption \ref{ass_n}(iv) holds, then $1$ is unstable in probability.
\end{enumerate}
\end{thm}

\bigskip
From this theorem we are able to derive the following propositions. The proof for each proposition follow the same argument as the analogous proposition in the previous section. For brevity, we merely state each proposition. Just as in the previous section, if the dynamics that are switched all have the same unstable vertex, then this vertex will also be unstable for the switched dynamic.This can be shown by setting $c_1=c_2=\ldots=c_n$ for the appropriate $V_{\cdot}^{-}$ function. The same reasoning my be applied to a stable vertex.

Just as in the previous section, we take the initial condition to be in $(0,1) \times I$.

\bigskip
\begin{prop}[Strategy 2 Dominant]\label{0as_n}
For the dynamic defined by Equation \eqref{rms_n}, if Assumptions \ref{ass_n}(i) and \ref{ass_n}(iv) holds, then $S(t) $ converges to 0 almost surely.  
\end{prop}

\bigskip
\begin{cor}[Strategy 1  Dominant]
If Assumptions \ref{ass_n}(ii) and \ref{ass_n}(iii) holds, then $ S(t) $ converges to 0 almost surely. 
\end{cor}

\bigskip
These conditions state the sample paths of the process tends to veer away from a vertex and evolve to the stable vertex. The initial condition does not have an influence on the dynamic.

\bigskip
\begin{prop}[Coordination]
For the dynamic defined by Equation \eqref{rms_n}, if Assumptions \ref{ass_n}(i) and \ref{ass_n}(iii) holds, then $S(t)$ is transient and converges to 0 or 1 almost surely.  
\end{prop}

\bigskip
The assumptions tell us that both vertices are stable, and as such the initial condition should affect the probabilities of converging to 0 and converging to 1. However, we are unable to determine these probabilities. 

\bigskip
\begin{prop}[Mixed Strategy  Dominant]
For the dynamic defined by Equation \eqref{rms_n}, if Assumptions \ref{ass_n}(ii) and \ref{ass_n}(iv) holds, then $S(t)$ is positive recurrent in a strict subinterval on $(0,1)$.  
\end{prop}

\bigskip
The vertices in these assumptions are unstable and thus we have that the process is recurrent.

\section{Conclusions}
We considered a new dynamic that consisted of independent switching between the underlying game payoffs and variances of the stochastic replicator process. Conditions for the switched stochastic replicator dynamic to be unstable or asymptotically stable for each vertex were derived. The combination of these conditions helped characterize the long-run behavior of the dynamic. The characterizations are an extension to the inequalities derived by Fudenberg and Harris \cite{FH92} that include the difference of the payoffs in a column of the payoff matrix, half the difference of the variances, and the \textit{comparison} of a fixed state transitioning to all other states. Since the switching indirectly perturbs the dynamic, each state is represented by an appropriately chosen, but not unique, constant. The difference of the constant for the fixed state and the constant of another state, multiplied by the transition rate from the fixed state to the other state, is the way we compare these two states. The sum of these terms is the \textit{comparison} of the fixed state to all other states. There is an extra term involving the variance, however, this term is negligible. 

Although the proofs focused on the case when the continuous Markov chain is in a two atom state space, the arguments have a very natural extension to the general finite atom state space. It may be possible to find better stochastic Lyapunov functions to obtain more precise conditions than the ones derived in this paper. However, since the inequalities are an extension of the criteria derived by Fudenberg and Harris, the conditions may be optimal. 

Finally, I believe that the assumptions given are exhaustive. Consider this example in a two atom state: if $\displaystyle - B_1 + \frac{ - \alpha - 1}{2} \sigma_1^2>0$, $\displaystyle  - B_2 + \frac{ - \alpha - 1}{2} \sigma_2^2 <0$, $\displaystyle - B_1 + \frac{ - \alpha - 1}{2} \sigma_1^2 = \Big| - B_2 + \frac{ - \alpha - 1}{2} \sigma_2^2\Big|$, $\displaystyle - B_1 + \frac{ \alpha - 1}{2} \sigma_1^2 =\Big| - B_2 + \frac{ \alpha - 1}{2} \sigma_2^2\Big|$, and $q_{12}=q_{21}$, then we are unable to saying anything about the dynamic 0. However, the two equalities imply that $\displaystyle \Big| - B_2 + \frac{ - \alpha - 1}{2} \sigma_2^2\Big| < \Big| - B_2 + \frac{ \alpha - 1}{2} \sigma_2^2\Big|$, which is impossible. Similar reasoning for other cases is how this hypothesis was derived.

\begin{acknowledgement*}
The author would like to thank Marc Harper for his wonderful comments.
\end{acknowledgement*}

\bibliography{ref_rms}                                                                                                                                               

\begin{thebibliography}{10}

\bibitem{aidoo2002protective}
Michael Aidoo, Dianne~J Terlouw, Margarette~S Kolczak, Peter~D McElroy, Feiko~O
  ter Kuile, Simon Kariuki, Bernard~L Nahlen, Altaf~A Lal, and Venkatachalam
  Udhayakumar.
\newblock Protective effects of the sickle cell gene against malaria morbidity
  and mortality.
\newblock {\em The Lancet}, 359(9314):1311--1312, 2002.

\bibitem{allison1954protection}
Anthony~C Allison.
\newblock Protection afforded by sickle-cell trait against subtertian malarial
  infection.
\newblock {\em British medical journal}, 1(4857):290, 1954.

\bibitem{A91}
W.J. Anderson.
\newblock {\em Continuous-Time {M}arkov Chains}.
\newblock Springer-Verlag, Berlin-Heidelberg, 1991.

\bibitem{AH05}
D.~Andersson and D.~Hughes.
\newblock Bacterial persistence and antibiotic resistance.
\newblock {\em Sci. STKE}, 2004(250):tw331, 2004.

\bibitem{BW11}
W.~Bai and P.~Wang.
\newblock Conditional markov chain and its application in economic time series
  analysis.
\newblock {\em Journal of applied econometrics}, 26:715--734, 2011.

\bibitem{C09}
S.~Choi.
\newblock Regime-switching univariate diffusion models of the short-term
  interest rate.
\newblock {\em Studies in Nonlinear Dynamics \& Econometrics}, 13(1):Article 4,
  2009.

\bibitem{cressman2003evolutionary}
Ross Cressman.
\newblock {\em Evolutionary dynamics and extensive form games}, volume~5.
\newblock the MIT Press, 2003.

\bibitem{DMT95}
D.~Down, S.~P. Meyn, and R.~L. Tweedie.
\newblock Exponential and uniform ergodicity of markov processes.
\newblock {\em Annals of Applied Probability}, 23(4):1671--1691, 1995.

\bibitem{FH92}
D.~Fudenberg and C.~Harris.
\newblock Evolutionary dynamics with aggregate shocks.
\newblock {\em Journal of Economic Theory}, 57(2):420--441, 1992.

\bibitem{FL12}
D.~Fudenberg and L~Imhof.
\newblock Phenotype switching and mutations in random environment.
\newblock {\em Bull Math Biol}, 74:399--421, 2012.

\bibitem{GS72}
I.~Gihman and A.~V. Skorohod.
\newblock {\em Stochastic differential equations}.
\newblock Springer-Verlag, New York, 1972.

\bibitem{GGXJ12}
A.~Gray, D.~Greenalgh, X.~Mao, and J.~Pan.
\newblock The {S}{I}{S} epidemic model with {M}arkovian switching.
\newblock {\em J. Math. Anal. Appl.}, 394:496--516, 2012.

\bibitem{GGMP12}
A.~Gray, D.~Greenhalgh, X.~Mao, and J.~Pan.
\newblock The {S}{I}{S} epidemic model with markovian switching.
\newblock {\em Journal of Mathematical Analysis and Applications},
  394:496--516, 2012.

\bibitem{H89}
J.~Hamilton.
\newblock A new approach to the economic analysis of non stationary time series
  and the business cycle.
\newblock {\em Econometrica}, 57(2):357--384, 1989.

\bibitem{HFV13}
M.~Harper, D.~Fryer, and A.~Vlasic.
\newblock Mean evolutionary dynamics for stochastically switching environments.
\newblock {\em http://arxiv.org/abs/1306.2373}, 2013.

\bibitem{hofbauer1998evolutionary}
Josef Hofbauer and Karl Sigmund.
\newblock {\em Evolutionary games and population dynamics}.
\newblock Cambridge University Press, 1998.

\bibitem{hofbauer2003evolutionary}
Josef Hofbauer and Karl Sigmund.
\newblock Evolutionary game dynamics.
\newblock {\em Bulletin of the American Mathematical Society}, 40(4):479, 2003.

\bibitem{I05}
I.~Imhof.
\newblock The long-run behavior of the stochastic replicator dynamics.
\newblock {\em Annals of Applied Probability}, 15(1B):1019--1045, 2005.

\bibitem{JW11}
J.~Janczura and R.~Weron.
\newblock Efficient estimation of markov regime-switching models: An
  application to electricity spot prices.
\newblock {\em Adv. Stat. Anal.}, 96:385--407, 2011.

\bibitem{RK80}
R.Z. Khasminskii.
\newblock {\em Stochastic stability of differential equations}.
\newblock Sijthoff and Noordhoff, Netherlands, 1980.

\bibitem{KZY07}
R.Z. Khasminskii, C.~Zhu, and G.~Yin.
\newblock Stability of regime-switching diffusions.
\newblock {\em Stochastic Process and Their Applications}, 117:1037--1051,
  2007.

\bibitem{KKBL05}
E.~Kussel, R.~Kishony, N.Q. Balabon, and S.~Leibler.
\newblock Bacterial persistence: A model of survival in changing environments.
\newblock {\em Genetics}, 169(4):1807--1814, 2005.

\bibitem{KL05}
E.~Kussell and S.~Leibler.
\newblock Phenotype diversity, population growth, and information in
  fluctuating environments.
\newblock {\em Science}, 309:2075--2078, 2005.

\bibitem{MY06}
X.~Mao and C.~Yuan.
\newblock {\em Stochastic Differential Equations with Markovian Switching}.
\newblock Imperial College Press, 2006.

\bibitem{MT93}
S.~P. Meyn and R.~L. Tweedie.
\newblock Stability of {M}arkovian {P}rocesses {I}{I}{I}: {F}oster-{L}yapunov
  criteria for continuous-time processes.
\newblock {\em Annals of Applied Probability}, 25(1):518--548, 1993.

\bibitem{MT93_1}
S.~P. Meyn and R.~L. Tweedie.
\newblock A survey of {F}oster-{L}yapunov techniques for general state space
  {M}arkov processes.
\newblock {\em Proc. Workshop Stochastic Stability and Stochastic
  Stabilization}, 1993.

\bibitem{moran1962statistical}
Patrick Alfred~Pierce Moran et~al.
\newblock The statistical processes of evolutionary theory.
\newblock {\em The statistical processes of evolutionary theory.}, 1962.

\bibitem{S89}
A.~V. Skorohod.
\newblock {\em Asymptotic methods in the theory of stochastic differential
  equations}.
\newblock American Mathematical Society, Moscow, 1989.

\bibitem{TDHS06}
Y.~Takeuchi, N.H. Du, N.T. Hieu, and K.~Sato.
\newblock Evolution of predator-prey systems described by a {L}otka-{V}olterra
  equation under random environment.
\newblock {\em J. Math. Anal. Appl.}, 323:938--957, 2006.

\bibitem{TCI07}
Ch. Taylor, J.~Chen, and Y.~Iwasa.
\newblock Cooperation maintained by fitness adjustment.
\newblock {\em Evolutionary Ecology Research}, 9:1023--1041, 2007.

\bibitem{weibull1995evolutionary}
J{\"o}rgen~W Weibull.
\newblock {\em Evolutionary game theory}.
\newblock The MIT press, 1995.

\bibitem{YZ07}
G.~Yin and Ch. Zhu.
\newblock Regularity and recurrence of switching diffusions.
\newblock {\em Jrl Syst Sci \& Complexity}, 20:273--283, 2007.

\end{thebibliography}
\bibliographystyle{plain}                                                                                                                                                                              
\nocite{*}

\end{document}